\documentclass[review]{interact}
\usepackage{amsmath,amssymb}
\usepackage{amsthm}

\usepackage{arydshln}
\usepackage{bm}
\usepackage{fancybox}
\usepackage{color}
\usepackage[caption=false]{subfig}
\usepackage[numbers,sort&compress]{natbib}
\usepackage[inline]{enumitem}
\usepackage{comment}
\usepackage{here}
\usepackage{algorithmic}
\usepackage{algorithm}
\bibpunct[, ]{[}{]}{,}{n}{,}{,}
\makeatletter
\def\NAT@def@citea{\def\@citea{\NAT@separator}}
\makeatother

\theoremstyle{plain}
\newtheorem{ass}{Assumption} 

\newtheorem{thm}{Theorem}
\newtheorem{lemma}{Lemma}

\newtheorem{remark}{Remark}

\newtheorem{deff2}{Definition}

\usepackage{color}
\definecolor{UniBlue}{RGB}{30,30,150} 
\definecolor{rits2}{RGB}{174,10,20}
\definecolor{sky}{RGB}{50,100,150}
\definecolor{tnmborange}{RGB}{245, 86, 90}

\def\revrev#1{\textcolor{black}{#1}}

\usepackage{dsfont}
\usepackage{ascmac}
\usepackage{fancybox}
\usepackage{url}
\usepackage{multirow}
\usepackage{charter}

\begin{document}


\articletype{}

\title{
Data-Driven Sub-Optimal LQ Regulator for Linear Input-Delay Systems based on Informativity
}

\author{
\name{Kohei~Ayaka\textsuperscript{a}, 
\thanks{CONTACT  K.~Takaba. Email: ktakaba@fc.ritsumei.ac.jp} 
Takumi~Namba\textsuperscript{a} and Kiyotsugu Takaba\textsuperscript{a}}
\affil{\textsuperscript{a} Department of Electrical and Electronic Engineering, Ritsumeikan University, 1-1-1 Noji-Higashi, Kusatsu, Shiga, Japan}
}

\maketitle

\begin{abstract}
This paper proposes a novel informativity-based data-driven synthesis method for a sub-optimal linear quadratic (LQ) regulator for linear input-delay systems from noisy input-state data.
Exploiting the augmented state structure of input-delay systems with a known delay length, we derive a linear matrix inequality (LMI) condition for the data-driven synthesis of the augmented state-feedback controller that achieves a prescribed LQ performance level for every plant model consistent with the data.
 The proposed LMI condition enables efficient controller synthesis via convex optimization.  
 Numerical simulations demonstrate the effectiveness of the proposed method.
 The trade-off between the achievable LQ performance and the uncertainty in the data is also clarified through a numerical example. 
\end{abstract}


\section{Introduction}
\label{sec:introduction}

With the recent development of data science technologies, data-driven control, which designs a controller directly from plant response data, has been attracting considerable attention~\cite{sankou7,sankou1,sankou6,subopt_vanwaarde_2020,sankou2}. 
In contrast to model-based control, data-driven control can directly handle plant uncertainties inherent in the data and can reduces the time required for controller design by skipping the system identification process. 
Among various data-driven approaches, several studies have addressed direct data-driven control based on \emph{data informativity} proposed by van~Waarde~{\it et al.}\cite{sankou1,sankou2,subopt_vanwaarde_2020,sankou6}. 
The data informativity characterizes whether the available data are sufficient for achieving a given control objective, such as stabilization. 
In particular, they derived a linear matrix inequality (LMI) condition for the data-driven synthesis of a quadratically stabilizing controller from noisy data~\cite{sankou2}.

Linear input-delay systems arise in various applications, including process control, networked control, remote control, etc. 
If the delay length is known {\it a priori}, it is a common technique to represent the input-delay systems by an augmented state equation, in which the state vector is extended to include the delayed inputs~(see e.g.\:\cite{sankou3,sankou4}). 
Even in data-driven control, this structure should be explicitly exploited in controller synthesis to achieve the desired control performance.
Therefore, we consider a data-driven sub-optimal LQR synthesis for linear input-delay systems based on the augmented state structure. 

In this paper, we derive an LMI condition for the data-driven synthesis of the \emph{sub-optimal} linear quadratic (LQ) regulator for a linear input-delay system.
While stabilization is a fundamental requirement, practical control design also demands guarantees on control performance, 
such as fast and smooth state convergence and moderate control effort, etc.
A common method for meeting such a requirement is to formulate the optimal LQ regulator (LQR) synthesis by introducing a quadratic performance index.
Nevertheless, the optimal LQR solution cannot be obtained when the exact plant model is not available due to the response data being corrupted by noise.
To remedy this difficulty, this paper considers the data-driven \emph{sub-optimal} LQR synthesis, where the performance index is guaranteed to be below a prescribed level for every plant model consistent with the noisy data.
This approach enables us to synthesize a stabilizing controller that explicitly accounts for the trade-off between control performance and robustness against the data uncertainty.

There have been several works related to this paper. 
For linear systems without input delays, van Waarde and Mesbahi~\cite{subopt_vanwaarde_2020} and van Waarde, Camlibel, and Meshbahi~\cite{sankou6} addressed data-driven synthesis of the sub-optimal LQR regulator in the noiseless and noisy data cases, respectively. 
This paper extends the latter work to linear input-delay systems with a known delay length.

In the conference paper~\cite{sicefes_ayaka}, we derived an LMI condition for data-driven quadratic stabilization of a linear input-delay system by exploiting its augmented state structure.  
However, the analysis in \cite{sicefes_ayaka} focuses only on stabilization and provides no guarantees on control performance. 
To overcome this drawback, we introduce a quadratic performance index and devise an LMI-based data-driven method for synthesizing the augmented state feedback controller that guarantees a given sub-optimal LQ performance level. 

R.-Escobedo, Fridman, and Schiffer~\cite{sankou5} 
derived sufficient LMI conditions for data-driven synthesis of a state-feedback stabilizing controller for a linear system with unknown time delays based on the robust control theory. 
The main difference between this paper and the reference~\cite{sankou5} is that we synthesis the sub-optimal LQ regulator in the form of \emph{augmented state-feedback}, including the delayed input values, rather than the ordinary state-feedback.  
This is because we exploit the augmented state structure of the input-delay system with the known delay length.
Nevertheless, our LMI condition is \emph{compact} in the sense that it is derived using the input-state data of the original input-delay state equation, not those of the augmented state equation.

In view of the above discussion, we summarize the contributions of this paper as follows. 
Firstly, the proposed method enables us to obtain a feedback controller that achieves the guaranteed LQ performance for every plant model consistent with the data, even when the data is corrupted by noise.
Secondly, the proposed controller achieves a guaranteed LQ performance by exploiting the augmented state structure under the assumption that the delay length is known {\it a priori}.
Thirdly, we derive an LMI condition for data-driven synthesis of the sub-optimal LQ regulator.  
The desired controller can be obtained efficiently by solving the LMI.  This LMI condition also facilitates the analysis of the trade-off between the achievable LQ performance and the plant uncertainty inherent in the data.

The organization of this paper is as follows.
In Section~\ref{sect:2}, we introduce a data-based system description and the notion of data informativity. 
Then, we formulate the sub-optimal LQR problem for linear input-delay systems.  
We review the model-based solution to the sub-optimal LQR problem in Section~\ref{sect:3}. 
Some formulas in this section are also useful for data-driven control. 
In Section~\ref{sect:4}, we derive an LMI condition for data informativity for the sub-optimal LQR synthesis, as the main result of this paper.
Section~\ref{sect:5} presents the simulation results that demonstrate the effectiveness of the proposed method. 
Concluding remarks are given in Section~\ref{sect:6}.

\section{Problem Formulation}
\label{sect:2}

\subsection{Data-based System Description}
\label{sec:2-1}

Throughout this paper, we consider a linear discrete-time system with an input delay $d$ defined by the nominal plant model
\begin{align}\label{ids}
    x_{t+1} = A x_t + B u_{t-d}, 
\end{align}
where $x_t \in \mathbb{R}^{n}$ and $u_t \in \mathbb{R}^{m}$ are the state and control input, respectively. 
The matrix pair $(A, B)$ is unknown to the designer. 
Instead, the input-state response data $(x_t,\:u_{t-d})$, $t=t_0,\dots, t_0+T-1$ is given as well as the terminal state $x_{t_0+T}$.
We assume that the response data is corrupted by process noise $w_t \in \mathbb{R}^{n}$, and is represented as
\begin{align}
  x_{t+1} = A x_t + B u_{t-d} + w_t, 
\ \ \: t=t_0,t_0+1,\dots, t_0+T-1.
\label{eq:se_delay_noise}
\end{align}
\begin{ass} \label{assum:delay_length}
The state dimension $n$ and the delay length $d$ are known a priori.
\end{ass}

One possible scenario for the above assumption is 
a situation where the delay length $d$ is obtained from a transient response, such as a step response, starting from the steady state.

We form the data set $\mathcal{D}=(X_+,X_-,U_-^d)$, where 
$X_+$, $X_-$, and $U^d_-$ are the data matrices defined by 
\begin{equation*}
\begin{aligned}
U_-^d &= \begin{bmatrix} u_{t_0-d} & u_{t_0-d+1} & \cdots & u_{t_0-d+T-1} \end{bmatrix}, \\
X_- &= \begin{bmatrix} x_{t_0} & x_{t_0+1} & \cdots & x_{t_0+T-1} \end{bmatrix}, \\
X_+ &= \begin{bmatrix} x_{t_0+1} & x_{t_0+2} & \cdots & x_{t_0+T} \end{bmatrix},\\
W_-&= \begin{bmatrix} w_{t_0} & w_{t_0+1} & \cdots & w_{t_0+T-1} \end{bmatrix}.
\end{aligned}
\end{equation*}

Then the equation (\ref{eq:se_delay_noise}) is equivalently rewritten as 
\begin{equation}
\label{all_sys}
X_+ = AX_- + BU_-^d + W_- 
.
\end{equation}

We make the following assumption on $W_-$ to characterize the noise property~\cite{sankou2}. 
\begin{ass}
For a given symmetric matrix 
$\Phi =\begin{bmatrix} \Phi_{11} & \Phi_{12} \\ \Phi_{12}^\top & \Phi_{22} \end{bmatrix}$ with $\Phi_{22}<0$, 
the matrix $W_-$ satisfies 
\begin{equation}
\label{noise_assumption}
\begin{bmatrix}
I \\
W_-^\top
\end{bmatrix}^\top
\begin{bmatrix}
\Phi_{11} & \Phi_{12} \\
\Phi_{12}^\top & \Phi_{22}
\end{bmatrix}
\begin{bmatrix}
I \\
W_-^\top
\end{bmatrix} \ge 0
.
\end{equation}
\end{ass}
By combining \eqref{noise_assumption} with \eqref{all_sys}, we get 
\begin{equation}
\left[\!\!
\begin{array}{c}
I_n \\ \hdashline 
A^\top \\
B^\top
\end{array}
\!\!\right]^\top
\begin{bmatrix}
\Psi_{11} & \Psi_{12} \\
\Psi_{12}^\top & \Psi_{22}
\end{bmatrix}
\left[\!\!
\begin{array}{c}
I_n \\ \hdashline 
A^\top \\
B^\top
\end{array}
\!\!\right]
\ge 0, 
\label{all_system}
\end{equation}
where
\begin{align}
\Psi = 
\begin{bmatrix}
\Psi_{11} & \Psi_{12} \\ \Psi_{12}^\top & \Psi_{22}
\end{bmatrix}
=
\begin{bmatrix}
I & X_+ \\
0 & -X_- \\
0 & -U_-^d
\end{bmatrix}
\begin{bmatrix}
\Phi_{11} & \Phi_{12} \\
\Phi_{12}^\top & \Phi_{22}
\end{bmatrix}
\begin{bmatrix}
I & X_+ \\
0 & -X_- \\
0 & -U_-^d
\end{bmatrix}^\top
.
\label{eq:Psi}
\end{align}
Define the set of all models consistent with the data $\mathcal{D}$ by 
\begin{eqnarray}
\Sigma_{\mathcal{D}}:=\{(A, B)\in{\mathbb{R}}^{n\times n}\times {\mathbb{R}}^{ n\times m} \mid (\ref{all_system})\text{ is satisfied.}\}\notag
\end{eqnarray}

The notion of data informativity proposed by van Waarde \emph{et al.}~\cite{sankou1} provides a set-theoretic characterization of the response data necessary and sufficient to achieve a desired control objective such as stabilization. 
Let $\mathcal{S}$ be a desired control specification. 
Define $\Sigma_{{\mathcal{S}},K}$ as the set of systems that satisfy $\mathcal{S}$ by a given controller $K$. 
Then, the data informativity of a control system is stated in ~Definition~\ref{def:informativity}.

\begin{deff2}
\label{def:informativity}
The data $\mathcal{D}$ is said to be informative for the control specification $\mathcal{S}$ if there exists a controller $K$ such that $\Sigma_{\mathcal{D}} \subseteq \Sigma_{{\mathcal{S}},K}$.
\end{deff2}

In particular, if we take ``stabilization" as the specification $\mathcal{S}$,  $\Sigma_{{\mathcal{S}},K}$ represents the set of all systems stabilized by $K$. The informativity in this case means that there exists a controller $K$ which stabilizes every system in $\Sigma_{\mathcal{D}}$.

\begin{figure}[H]
\begin{centering}
\includegraphics[width=8cm]{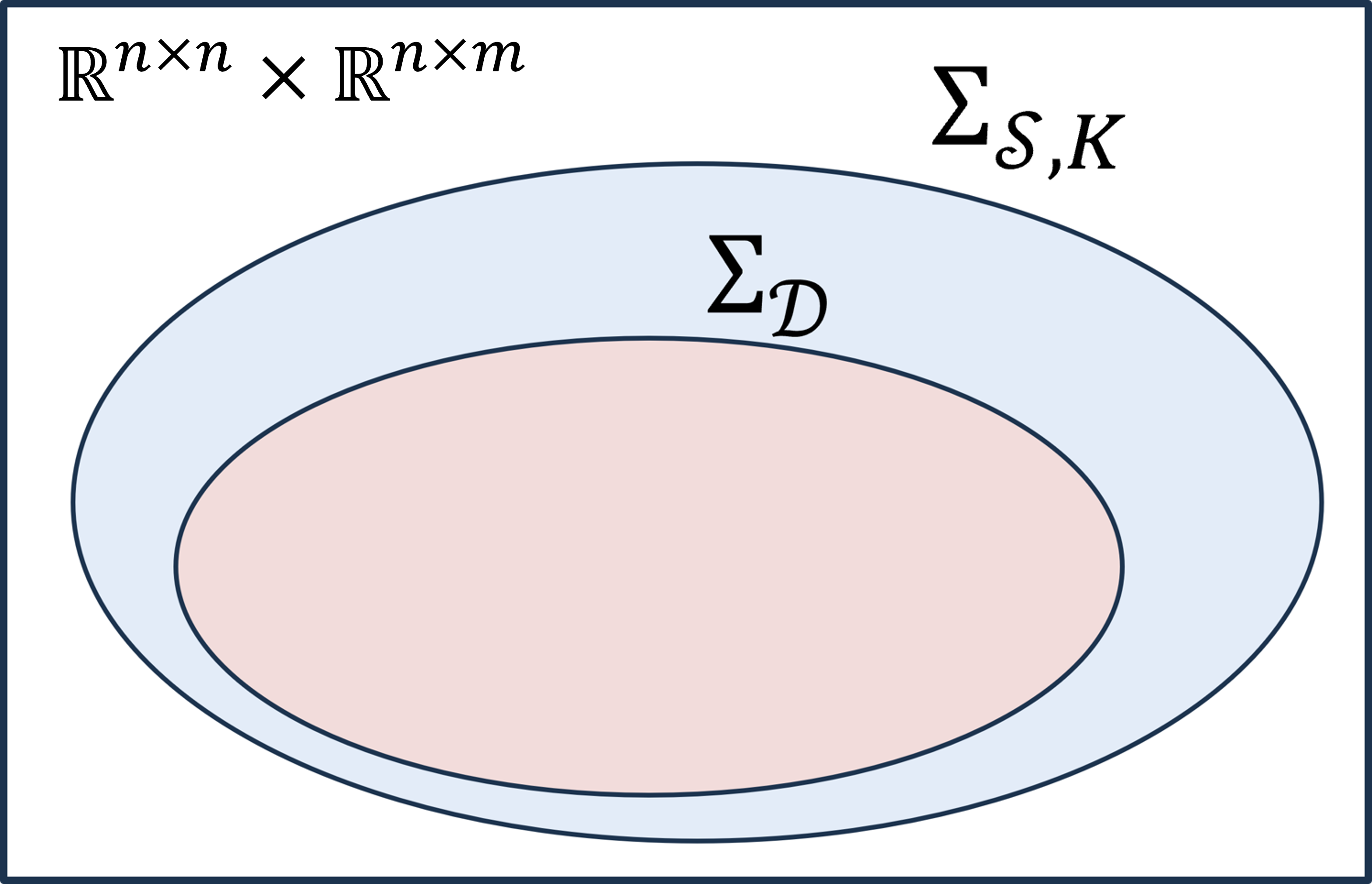}
\caption{Data informativity}
\end{centering}
\end{figure}

\subsection{Sub-Optimal Linear Quadratic Regulator Problem}\label{sect:2-2}
It is well known in the literature~\cite{sankou3,sankou4} that, by augmenting the state vector, the linear input-delay system (\ref{ids}) can be expressed by 
\begin{align*}
\left[\!\!\begin{array}{c}
x_{t+1} \\ \hdashline
u_{t-d+1} \\  
\vdots \\ 
u_{t-1} \\ \hdashline
u_t
\end{array}\!\!\right]
=&
\left[
\begin{array}{c:ccc:c}
A & B & 0 & \cdots & 0 \\ \hdashline
0 & 0 & I_m & \cdots & 0 \\
\vdots & \vdots & \vdots & \ddots & \vdots \\ 
0 & 0 & 0 & \cdots & I_m \\ \hdashline
0 & 0 & 0 & \cdots & 0
\end{array}
\right]
\left[\!\!
\begin{array}{c}
x_t \\ \hdashline
u_{t-d} \\ 
\vdots \\ 
u_{t-2} \\ \hdashline
u_{t-1}
\end{array}\right]  +
\left[\!\!\begin{array}{c}
0 \\ \hdashline
0 \\ 
\vdots \\ 
0 \\ \hdashline
I_m
\end{array}\!\!\right] u_{t}\:, 
\end{align*}
or in short hand, 
\begin{align}
X_{t+1} = \mathcal{A}X_t + \mathcal{B}u_t\, , 
\label{eq:5}
\end{align}
where 
\begin{align*}
X_t = 
\left[\!\!
\begin{array}{c}
x_t \\ \hdashline
u_{t-d} \\ 
\vdots \\ 
u_{t-2} \\ \hdashline
u_{t-1}
\end{array}\right]
,\ \ 
\mathcal{A}=
\left[
\begin{array}{c:ccc:c}
A & B & 0 & \cdots & 0 \\ \hdashline
0 & 0 & I_m & \cdots & 0 \\
\vdots & \vdots & \vdots & \ddots & \vdots \\ 
0 & 0 & 0 & \cdots & I_m \\ \hdashline
0 & 0 & 0 & \cdots & 0
\end{array}
\right],\ \ 
\mathcal{B}=
\left[\!\!\begin{array}{c}
0 \\ \hdashline
0 \\ 
\vdots \\ 
0 \\ \hdashline
I_m
\end{array}\!\!\right]
.
\end{align*}
It is observed from this equation that the system matrix has a sparse structure and that the plant model $\begin{bmatrix} A & B \end{bmatrix}$ appears only in the upper left block of the system matrix.

We introduce the quadratic performance index $J(X_0; u)$ as
\begin{align*}
J(X_0;\,u)= \sum_{t=0}^{\infty} 
\left(
{x}_t^\top Q_0 x_t 
+ \sum_{i=1}^{d} u_{t-d+i-1}^\top Q_i u_{t-d+i-1} +u_t^\top R u_t \right)
\end{align*}
where $Q_i$ ($i=0,1\dots,d$) and $R$ are positive semi-definite symmetric matrices of compatible dimensions.  
This performance index claims that the plant state $x_t$ should rapidly decay to the origin without excessive control effort. 
In particular, the penalties on the delayed inputs are also included in $J(X_0;u)$ to add generality.
The suboptimal linear quadratic regulator (LQR) problem is to synthesize a feedback controller 
that asymptotically stabilizes the closed-loop system and achieves the suboptimal performance
\begin{align}
J(X_0; u ) \le \gamma \|X_0\|^2 
\ \forall X_0 \in \mathbb{R}^{n+md}
\label{eq:subopt}
\end{align}
for a prescribed performance level $\gamma>0$.

By defining $Q=\mbox{\rm block-diag}(Q_0,Q_1,\dots,Q_d)$, $J(X_0;u)$ is rewritten 
in terms of the augmented state $X_t$ as 
\begin{align}
J(X_0; u) = \sum_{t=0}^\infty 
\left( X_t^\top Q X_t + u_t^\top R u_t \right)
.
\label{eq:cost}
\end{align}
Thus, the suboptimal LQR problem for the linear input-delay plant (\ref{ids}) reduces to that for the augmented plant (\ref{eq:5}), which is free of input delays.
Based on this observation, we apply the augmented state feedback controller
\begin{align}\label{fb_law}
u_t = 
KX_t=
K_0x_t+\sum_{i=1}^dK_iu_{t-d+i-1}
\end{align}
to the augmented state equation (\ref{eq:5})
, where 
$K=[\,K_0\:K_1\:\dots\:K_d\,]\in \mathbb{R}^{m\times (n+md)}$ is the feedback gain to be designed. 
Then, the closed-loop system is given by
\begin{align}
&
X_{t+1} = {\mathcal{A}_{cl}}X_t,
\ \ \ %
\mathcal{A}_{cl} = \mathcal{A}+\mathcal{B}K=
\left[
\begin{array}{cc:ccc}
\!\!A\!\! & \!B\! & \!0 & \cdots & 0 \!\! \\ \hdashline
\!\!0\!\! & \!0\! & \! I_m & \cdots & 0 \!\! \\
\!\!\vdots\!\! & \!\vdots\! & \!\vdots & \ddots & \vdots \!\! \\ 
\!\!0 \!\! & \!0\! &\! 0 & \cdots & I_m \!\! \\ \hdashline
\!\!K_0\!\! & \!K_1\! &\! K_2 & \cdots & K_d \!\!
\end{array}
\right].
\label{eq:cls} 
\end{align}
The feedback gain $K$ is called \emph{a sub-optimal feedback gain} if it stabilizes the closed-loop matrix $\mathcal{A}_{cl}$ and achieves (\ref{eq:subopt}) for a given $\gamma>0$.

\section{Model-Based Sub-optimal LQR Synthesis}
\label{sect:3}

In this section, we review the model-based synthesis of the sub-optimal LQ regulator under the following assumption.

\begin{ass}
\leavevmode\\[-1em]
\begin{itemize}
\item[(i)] $(A,B)$ is available for controller synthesis.
\item[(ii)] $(A,B)$ is stabilizable.
\end{itemize}
\end{ass}

The following lemma is standard in the LQ regulator theory (see e.g.\ \cite{sankou6}).

\begin{lemma}
Let a gain matrix $K$ and a positive constant $\gamma$ be given. 
If there exists a positive definite symmetric matrix $S$ satisfying 
\begin{align}
\mathcal{A}_{\mathrm{cl}}^{\top} S \mathcal{A}_{\mathrm{cl}}
- S + Q + K^{\top} R K &< 0,
\label{eq:lqr_ineq} \\
S &\le \gamma I_{n+md},
\label{eq:S_bound}
\end{align}
then $K$ is a sub-optimal feedback gain for $\gamma$, which stabilizes the closed-loop matrix $\mathcal{A}_{\mathrm{cl}}$ and achieves 
\begin{align}
J(X_0; u)
\le X_0^{\top} S X_0
\le \gamma \| X_0 \|^2\ \ \forall X_0 \in \mathbb{R}^{n+md}.
\label{eq:cost_bound}
\end{align}
\end{lemma}

Define $P = S^{-1}$.
By pre- and post-multiplying \eqref{eq:lqr_ineq} by $P$, we obtain 
\begin{align*}
- P \mathcal{A}_{\mathrm{cl}}^{\top} P^{-1} \mathcal{A}_{\mathrm{cl}} P
+ P - P Q P - P K^{\top} R K P > 0. 
\end{align*}

Since $P > 0$, applying the Schur complement formula to the above inequality yields
\begin{align}
\begin{bmatrix}
P & \mathcal{A}_{\mathrm{cl}} P \\
P \mathcal{A}_{\mathrm{cl}}^{\top}
& P - P Q P - P K^{\top} R K P
\end{bmatrix}
> 0.
\label{eq:lmi_basic}
\end{align}
A further application of the Schur complement formula, together with the change of variables $L=KP$, results in
\begin{align}
\begin{bmatrix}
P & \mathcal{A} P + \mathcal{B} L & 0 & 0 \\
(\mathcal{A} P + \mathcal{B}L)^{\top}
& P & P Q^{1/2} & L^{\top} R^{1/2} \\
0 & Q^{1/2} P & I_{n+p} & 0 \\
0 & R^{1/2} L & 0 & I_m
\end{bmatrix}
> 0.
\label{eq:lmi_extended}
\end{align}

In addition, since $P=S^{-1}>0$, (\ref{eq:S_bound}) is equivalent to
\begin{align}
\begin{bmatrix}
\gamma I_{n+md} & I_{n+md} \\
I_{n+md} & P
\end{bmatrix}
> 0.
\label{eq:lmi_gamma}
\end{align}

\newpage

\begin{thm}
Under Assumption~3, a necessary and sufficient condition for the existence
of a sub-optimal feedback gain $K$ for a given $\gamma>0$ is that there exist a symmetric matrix
$P \in \mathbb{R}^{(n+md)\times(n+md)}$ 
and a matrix $L \in \mathbb{R}^{m\times(n+md)}$ satisfying
\eqref{eq:lmi_extended} and \eqref{eq:lmi_gamma}.
In this case, such a feedback gain is given by
\begin{equation}
K = L P^{-1}. 
\end{equation} 
\end{thm}

\section{Data-Driven Sub-optimal LQR Synthesis}
\label{sect:4}

Consider the data-driven sub-optimal LQR synthesis for the plant (\ref{ids}) under the situation where the plant model $(A,B)$ is unknown.
We make the following assumption in addition to Assumptions~1 and~2.

\begin{ass}
\label{ass:data_set}
The data set $\mathcal{D} = (X_+, X_-, U_-^d)$ defined in Sub-section~2.1 is available for controller synthesis.
\end{ass}

Define $\Sigma_{{\mathrm{subopt}},{\mathcal{P}},K}$ as the set of all $(A,B)$ for which (\ref{eq:lqr_ineq}) and (\ref{eq:S_bound}) are satisfied for a positive definite matrix ${\mathcal{P}}$ and an augmented state feedback gain $K=[K_0\ K_1\:\cdots\:K_d]\in {\mathbb{R}}^{m\times (n+md)}$. 
In view of Definition~\ref{def:informativity}, we define the data informativity for sub-optimal LQR control of the linear input-delay system as follows.

\begin{deff2}
Let a performance level $\gamma$ be given. 
The data ${\mathcal{D}}=(X_{+},X_{-},U_{-}^d)$ is said to be informative for sub-optimal LQR control if there exists a feedback gain $K=
\begin{bmatrix}
K_0 & K_1 & \cdots & K_d
\end{bmatrix}\in{\mathbb{R}}^{m \times (n+md)}$ and a positive definite matrix ${\mathcal{P}}\in {\mathbb{R}}^{(n+md)\times (n+md)}$  satisfying $\Sigma_{\mathcal{D}}\subseteq\Sigma_{{\mathrm{subopt}}, {\mathcal{P}},K}$, i.e., (\ref{eq:lqr_ineq}) and (\ref{eq:S_bound}) are satisfied for all $(A,B)\in\Sigma_{\mathcal{D}}$. 
\end{deff2}

To derive an LMI condition for the informativity for sub-optimal LQR,  we first recap the matrix inequality (\ref{eq:lmi_basic}).
Let $\Pi = P - P Q P - P K^\top R K P$. Then, by the Schur complement, 
the inequality~\eqref{eq:lmi_basic} is equivalent to
\begin{equation}
\label{eq:6.3}
\Pi > 0 \quad \text{and} \quad 
P - \mathcal{A}_{cl} P \Pi^{-1} P {\mathcal{A}_{cl}}^\top > 0.
\end{equation}
Notice that the closed-loop matrix $\mathcal{A}_{cl}$ can be expressed as  
\begin{equation}
\label{eq:6.4}
\mathcal{A}_{cl} = 
\left[
\begin{array}{c:c}
Z^\top & 0 \\
0 & I_{m(d-1)} \\
K_{0:1} & K_{2:d}
\end{array}
\right]
,\ \ \  
Z= [\,A\ \:B\,]^\top,
\notag
\end{equation}
Accordingly, we partition $P$ as
\begin{equation}
\label{eq:6.6}
\setlength{\dashlinedash}{2pt} 
\setlength{\dashlinegap}{2pt}  
\setlength{\arrayrulewidth}{0.3pt} 
\begin{aligned}
P &=
\left[
\begin{array}{cc:cc}
P_{00} & P_{10}^\top & \cdots & P_{d0}^\top \\
P_{10} & P_{11} & \cdots & P_{d1}^\top \\ \hdashline
\vdots & \vdots & \ddots & \vdots \\ 
P_{d0} & P_{d1} & \cdots & P_{dd}^\top
\end{array}
\right]
=: 
\left[ \begin{array}{c:c} \hat{P}_a & \hat{P}_b^\top \\ \hdashline \hat{P}_b & \hat{P}_c \end{array}\right]=:
\left[
\begin{array}{c}
P_a \\ \hdashline
P_b
\end{array}
\right]
\end{aligned}
.\notag
\end{equation}
Then the second inequality in \eqref{eq:6.3} can be rewritten as
\begin{equation}
\label{eq:6.5}
\begin{bmatrix} 
\hat{P}_a & \hat{P}_b \\ \hat{P}_b^\top & \hat{P}_c 
\end{bmatrix} 
> 
\begin{bmatrix}
Z^\top {Q}_1 Z & Z^\top {Q}_2^\top \\
{Q}_2 Z & {Q}_3 \end{bmatrix},
\end{equation}
where 
\begin{align}
Q_1 &= P_a \Pi^{-1} {P_a}^\top, \notag\\
Q_2 &= 
\begin{bmatrix}
P_b \Pi^{-1} {P_a}^\top \\
(K_{0:1} P_a + K_{2:d} P_b) \Pi^{-1} {P_a}^\top
\end{bmatrix}, \notag\\
Q_3 &=
\begin{bmatrix}
P_b \Pi^{-1} {P_b}^\top & P_b \Pi^{-1} ({P_a}^\top K_{0:1}^\top + {P_b}^\top K_{2:d}^\top) \\
(K_{0:1} P_a + K_{2:d} P_b) \Pi^{-1} {P_b}^\top & (K_{0:1} P_a + K_{2:d} P_b) \Pi^{-1} ({P_a}^\top K_{0:1}^\top + {P_b}^\top K_{2:d}^\top)
\end{bmatrix}\notag
\end{align}

Since the plant model consistent with the data $\mathcal{D}$ is not unique, we need to achieve the LQ sub-optimality for all such plant models. 
For this purpose, we will derive a necessary and sufficient condition for (\ref{eq:6.5}) to hold for every $Z=[A\ B]^\top\in \Sigma_{\mathcal{D}}$.

To derive such a condition, we devise a modified version of the matrix S-lemma~\cite{sankou2} by expanding the block size of the matrices. 

\begin{lemma}[Expanded matrix S-lemma]
\label{lem:expanded_S_lemma}
Define 
\begin{align*}
\!\!\!\!\!\!%
{\mathcal{S}}_N := \left\{ Z \in {\mathbb{R}}^{m \times n} \;\middle|\ \ %
\begin{bmatrix}
I_n \\ Z
\end{bmatrix}^\top
\begin{bmatrix}
N_a & N_b^\top \\
N_b & N_c
\end{bmatrix}
\begin{bmatrix}
I_n \\ Z
\end{bmatrix} \ge 0
\right\}
\subset {\mathbb{R}}^{m\times n}
\end{align*}
for a symmetric matrix  
$
N = \begin{bmatrix} 
N_a & N_b^\top \\
N_b & N_c 
\end{bmatrix} 
\in {\mathbb{R}}^{(n+m) \times (n+m)}
$
satisfying 
$
N_c \preceq 0$ and $\ker N_c \subseteq \ker N_b^\top
$.
Let also symmetric matrices $P$ and $Q$  be given as 
\begin{align*}
& P = \begin{bmatrix}
P'_a & {P'_b}^\top \\
P'_b & P'_c
\end{bmatrix}\in {\mathbb{R}}^{(n+\ell) \times (n+\ell)}, 
\ \ 
Q = \begin{bmatrix}
Q_a & Q_b^\top \\
Q_b & Q_c
\end{bmatrix}\in {\mathbb{R}}^{(m+\ell) \times (m+\ell)},\ \ 
Q_a \ge 0.
\end{align*}
Then, 
\begin{align}
\label{eq:8}
    \begin{bmatrix}
P'_a & {P'_b}^\top \\
P'_b & P'_c
\end{bmatrix}
>
\begin{bmatrix}
Z^\top Q_a Z & Z^\top Q_b^\top \\
Q_b Z & Q_c
\end{bmatrix}
\quad \forall Z \in \mathcal{S}_N
\end{align}
holds if and only if there exist constants $\alpha \geq 0$ and $\varepsilon > 0$ satisfying
\begin{align}
\label{eq:9}
\begin{bmatrix}
P'_a - \varepsilon I_n & 0 & {P'_b}^\top \\
0 & -Q_a & -Q_b^\top \\
P'_b & -Q_b & P'_c - Q_c-\varepsilon'I_\ell
\end{bmatrix}
- \alpha
\begin{bmatrix}
N_a & N_b^\top & 0 \\
N_b & N_c & 0 \\
0 & 0 & 0
\end{bmatrix}
\ge 0.
\end{align}
\end{lemma}
\begin{proof} 
See Appendix~A.
\end{proof}

By Lemma~\ref{lem:expanded_S_lemma}, the inequality (\ref{eq:6.5}) is satisfied for every 
$[A\ B]\in \Sigma_{\mathcal{D}}$ if and only if there exist constants $\alpha \ge 0$, $\varepsilon, \varepsilon' > 0$ such that
\begin{equation}
\label{eq:6.12_english}
\begin{bmatrix}
{\hat{P}}_a - \varepsilon I_n & 0 & {\hat{P}}_b^\top \\
0 & -Q_1 & -Q_2^\top \\
{\hat{P}}_b & -Q_2 & {\hat{P}}_c - Q_3 - \varepsilon' I_{md}
\end{bmatrix} 
- \alpha 
\begin{bmatrix}
\Psi_a & \Psi_b^\top & 0 \\
\Psi_b & \Psi_c & 0 \\
0 & 0 & 0
\end{bmatrix} \ge 0.
\end{equation}
We further partition $\hat{P}_b$ and $\hat{P}_c$ as
\begin{equation}
\label{eq:L_def_english}
{\hat{P}}_b:=\left[\begin{matrix}{\hat{P}}_{b0}\\{\hat{P}}_{b1}\\\end{matrix}\right], {\hat{P}}_c:=\left[\begin{matrix}{\hat{P}}_{c00}&{{\hat{P}}_{c10}}^\top\\{\hat{P}}_{c10}&{\hat{P}}_{c11}\\\end{matrix}\right]. \notag
\end{equation}
Then, the inequality \eqref{eq:6.12_english} is expressed as
\begin{equation}
\label{eq:original_LMI_english}
\begin{bmatrix}
{\hat{P}}_a - \varepsilon I_n - \alpha \Psi_a & -\alpha \Psi_b^\top & {\hat{P}}_{b0}^\top & {\hat{P}}_{b1}^\top \\
-\alpha \Psi_b & -\alpha \Psi_c & 0 & 0 \\
{\hat{P}}_{b0} & 0 & {\hat{P}}_{c00} - \varepsilon' I_{m(d-1)} & {\hat{P}}_{c10} \\
{\hat{P}}_{b1} & 0 & {\hat{P}}_{c10} & {\hat{P}}_{c11} - \varepsilon' I_m
\end{bmatrix}
-
\begin{bmatrix}
0 \\ P_a \\ P_b \\ L
\end{bmatrix} \Pi^{-1}
\begin{bmatrix}
0 \\ P_a \\ P_b \\ L
\end{bmatrix}^\top
\ge 0. \notag
\end{equation}
Noting $\Pi = P - P Q P - L^\top R L>0$, repeated application of the Schur complement formula to this inequality leads to the LMI
\begin{equation}
\label{eq:main_LMI}
\begin{bmatrix}
{\hat{P}}_a - \varepsilon I_n - \alpha \Psi_a & -\alpha \Psi_b^\top & {\hat{P}}_{b0}^\top & {\hat{P}}_{b1}^\top & 0 & 0 & 0 \\
-\alpha \Psi_b & -\alpha \Psi_c & 0 & 0 & P_a & 0 & 0 \\
{\hat{P}}_{b0} & 0 & {\hat{P}}_{c00} - \varepsilon' I_{m(d-1)} & {\hat{P}}_{c10}^\top & P_b & 0 & 0 \\
{\hat{P}}_{b1} & 0 & {\hat{P}}_{c10} & {\hat{P}}_{c11} - \varepsilon' I_m & L & 0 & 0 \\
0 & P_a^\top & P_b^\top & L^\top & P & L^\top R^{1/2} & P Q^{1/2} \\
0 & 0 & 0 & 0 & R^{1/2} L & I_m & 0 \\
0 & 0 & 0 & 0 & Q^{1/2} P & 0 & I_{n + md}
\end{bmatrix} \ge 0.
\end{equation}

Regarding the performance index, similarly to the model-based synthesis case, $J$ can be upper-bounded by
\begin{equation}
\label{eq:performance_index_english}
\begin{bmatrix}
\gamma I_{n+p} & I_{n+p} \\
I_{n+p} & P
\end{bmatrix} > 0.
\end{equation}

Consequently, we obtain the following theorem as the main result of this paper.

\begin{thm}
\label{thm:main_result_LQR}
Given the data $\mathcal{D} = (X_+, X_-, U_-^d)$ and a constant $\gamma > 0$, assume that there exist a symmetric matrix $P \in \mathbb{R}^{(n+md) \times (n+md)}$, a matrix $L \in \mathbb{R}^{n \times (n+md)}$, and scalar constants $\alpha \ge 0$, $\varepsilon > 0$, $\varepsilon' > 0$ satisfying the LMIs \eqref{eq:main_LMI} and \eqref{eq:performance_index_english}.
Then, the data $\mathcal{D}$ is informative for the sub-optimal LQR control, namely there exists a feedback gain $K$ that stabilizes all models in $\Sigma_{\mathcal{D}}$ and achieves the sub-optimality 
\[
J(X_0;u) \le \gamma \| X_0 \|^2\ \:\forall X_0 \in \mathbb{R}^{n+md}
\]
for all models in $\Sigma_{\mathcal{D}}$. 
In this case, one such feedback gain is given by
\begin{equation}
\label{eq:K_feedback_LQR}
K = 
\begin{bmatrix}
K_0 & K_1 & \cdots & K_d
\end{bmatrix}
= L P^{-1}.
\end{equation}
\end{thm}

\begin{remark} 
The upper-left $5\times 5$ block of the LMI \eqref{eq:main_LMI} implies 
\begin{align*}
\begin{bmatrix}
{\hat{P}}_a - \varepsilon I_n - \alpha \Psi_a & -\alpha \Psi_b^\top & {\hat{P}}_{b0}^\top & {\hat{P}}_{b1}^\top & 0   \\
-\alpha \Psi_b & -\alpha \Psi_c & 0 & 0 & P_a  \\
{\hat{P}}_{b0} & 0 & {\hat{P}}_{c00} - \varepsilon' I_{m(d-1)} & {\hat{P}}_{c10}^\top & P_b  \\
{\hat{P}}_{b1} & 0 & {\hat{P}}_{c10} & {\hat{P}}_{c11} - \varepsilon' I_m & L  \\
0 & P_a^\top & P_b^\top & L^\top & P 
\end{bmatrix} \ge 0.
\end{align*}
Since this is nothing but the stabilization condition in Theorem~1 in the reference~\cite{sicefes_ayaka}, 
we can say that the above theorem constitutes an extension of the authors' previous work~\cite{sicefes_ayaka}.
\end{remark}

Notice that the matrix inequalities \eqref{eq:main_LMI} and \eqref{eq:performance_index_english} are affine with respect to $\gamma$. 
Hence, considering $\gamma$ as one of the decision variables, we can compute the achievable performance level $\gamma_{\rm opt}$ by solving the following convex optimization problem.

\begin{equation}
\label{eq:gamma_opt_english}
\text{(CP)}\inf_{P, L,  \alpha, \varepsilon, \varepsilon', \gamma} \gamma \quad \text{subject to \eqref{eq:main_LMI}, \eqref{eq:performance_index_english}}, \alpha\ge 0, \varepsilon>0, \varepsilon'>0, \gamma>0. \notag
\end{equation}

We conclude this section by summarizing the algorithm for the data-driven synthesis of a sub-optimal LQR based on Theorem~\ref{thm:main_result_LQR}.

\begin{algorithm}[H]
\caption{sub-optimal LQR Synthesis}
	\label{alg2}
\textbf{Given:\quad} $n, m, d, {\mathcal{D}}= (X_+,X_-, U_-^d) $ 
\begin{enumerate}
\renewcommand{\labelenumi}{\texttt{\arabic{enumi}:~}}
\item{Determine $\Phi_{11}=\Phi_{11}^\top$, $\Phi_{12}$, and $\Phi_{22}=\Phi_{22}^\top\prec0$ based on the assumed property of the noise $w$.}
\item{Compute the matrix $\Psi$ in \eqref{eq:Psi}.}
\item{Solve the problem (CP) by an existing convex optimization algorithm.}
\item{Obtain the desired feedback gain $K=\begin{bmatrix}
K_0 & K_1 & \cdots & K_d
\end{bmatrix}$ by (\ref{eq:K_feedback_LQR}).}
\end{enumerate}
\end{algorithm}

\section{Numerical Example}
\label{sect:5}
In this section, we verify the effectiveness of the proposed method through numerical examples.
We conducted all numerical computations and simulations using MATLAB R2023b and MOSEK 10.0.43.
\subsection{Simulation setup}
Consider the linear input-delay system (1) with the \emph{true} model 
\begin{equation}
A_s=
\left[\begin{matrix}
1.3&0.5\\ 
0&1.2 
\end{matrix}\right],\quad
B_s=
\left[\begin{matrix}
1\\ 
1
\end{matrix}\right]\notag.
\end{equation}
We assume that the delay length $d=4$ has been identified \emph{a priori} by some means.

We collect the data ${\mathcal{D}}=(X_+,X_-,U_-^d)$ with $T=10$ by injecting the sinusoidal input $u_t=5 \sin 10t$ to 
\begin{eqnarray}
x_{t+1} = A_sx_t + B_su_{t-d}+w_t, \notag
\end{eqnarray}
where the noise $w_t$ is a pseudo-random number sequence generated by 
a zero-mean white Gaussian noise sequence with covariance $1.0\times10^{-2} I_2$.
The obtained noisy data trajectories along with the input signal are depicted in Figure~\ref{fig:sim_data}.

\begin{figure}[!bth]
\begin{centering}
\includegraphics[width=10cm]{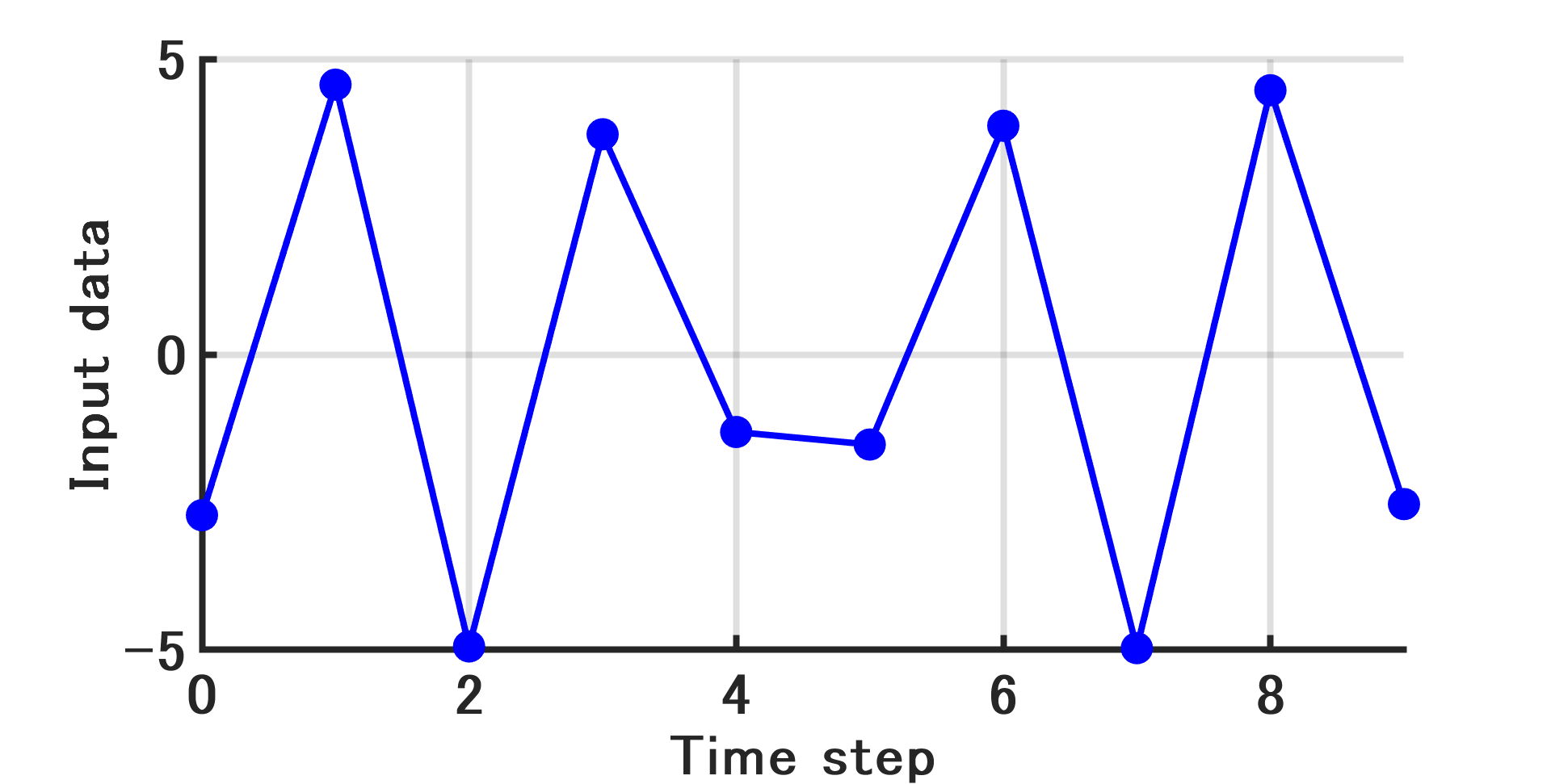}\\
(a) Input data 

\medskip

\includegraphics[width=10cm]{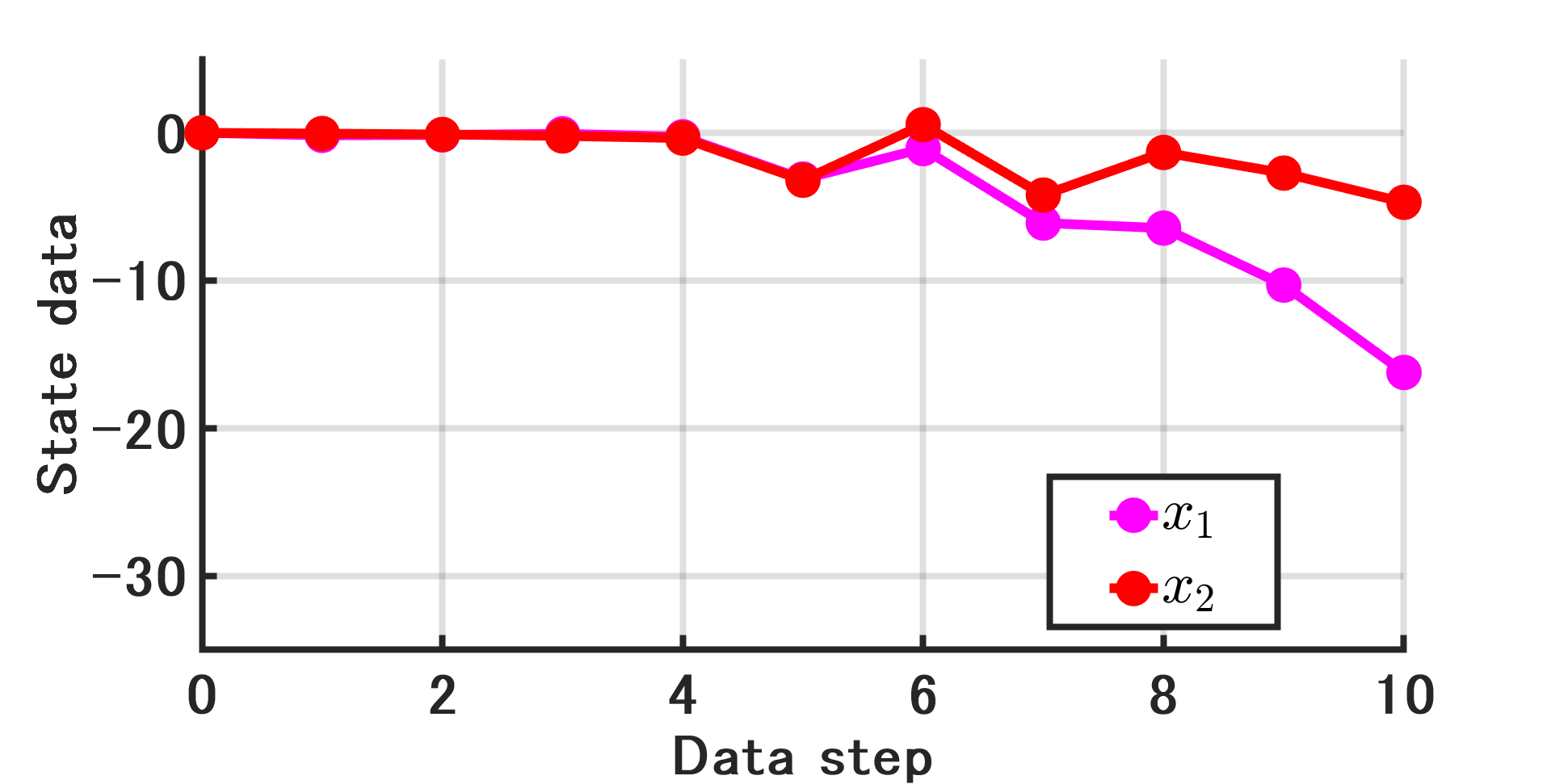}\\
(b) State data
\caption{Input-state data ${\mathcal{D}}$}
\label{fig:sim_data}
\end{centering}
\end{figure}

As for the inequality (\ref{noise_assumption}), we set 
\begin{align*}
\Phi = \begin{bmatrix} \sigma^2 T I_2 & 0 \\ 0 & - I_{10} \end{bmatrix}.
\end{align*}
This implies that the covariance of the noise $w_t$ is not greater than $\sigma^2 I_2$. 
Note that we have executed Algorithm~1 and  Algorithm~2 using the SDP solver in MOSEK to obtain a feasible solution. 

In this numerical example, we set
\[
Q_0 = 1.0\times10^{-4}\, I_2,\quad Q_1=Q_2=Q_3=Q_4 = 1.0\times 10^{-4},
\quad 
R = 3.0\times10^{-4} .
\]
We selected these weight values to ensure numerical stability.

\subsection{Noise bound $\sigma$}
With the above weighting matrices, the range of $\sigma$ for which the convex program (CP) is feasible and $\Sigma_{\mathcal{D}} \neq \emptyset$ is
\[
0.0875 \le \sigma \le 0.1256 .
\]
 
To examine the effect of $\sigma$ on the control performance,  
we synthesize the sub-optimal LQ regulator for three different values:
$\sigma = 0.0875,\ 0.1000,\ 0.1256$.
We solve the convex optimization problem (CP) for these values and obtain the achievable performance level $\gamma_{\rm opt}$ and the corresponding feedback gain $K$ as summarized in Table~\ref{tab:K_gamma_matrix}.
Note that $\gamma_{\rm opt}$ is the optimum of the optimization problem (CP).

\begin{table}[H]
\small
\centering
\caption{Feedback gain $K$ and achievable performance level $\gamma_{\rm opt}$}
\label{tab:K_gamma_matrix}
\begin{tabular}{|c|c|c|}
\hline
$\sigma$ &  Feedback gain $K$ & $\gamma_{\rm opt}$ \\
\hline
$0.0875$ &
$\begin{bmatrix}
-1.2695 & -3.0682 & -3.1090 & -2.2355 & -1.5846 & -1.0983
\end{bmatrix}$ &
$4.1162 \times 10^{-2}$ \\[1ex]
$0.1000$ &
$\begin{bmatrix}
-3.7490 & -6.1393 & -6.7617 & -4.5652 & -2.9629 & -1.8011
\end{bmatrix}$ &
1.0296 \\[1ex]
$0.1256$ &
$\begin{bmatrix}
-8.1792 & -10.7806 & -12.5874 & -8.1347 & -4.9126 & -2.6109
\end{bmatrix}$ &
19.3357 \\
\hline
\end{tabular}
\end{table}

In all these cases, the eigenvalues of the closed-loop matrix $\mathcal{A} + \mathcal{B}K$
for $(A_s,B_s)$ lie inside the unit disc, and hence each controller stabilizes the closed-loop system.

\begin{figure}[H]
\begin{centering}
\includegraphics[width=13cm]{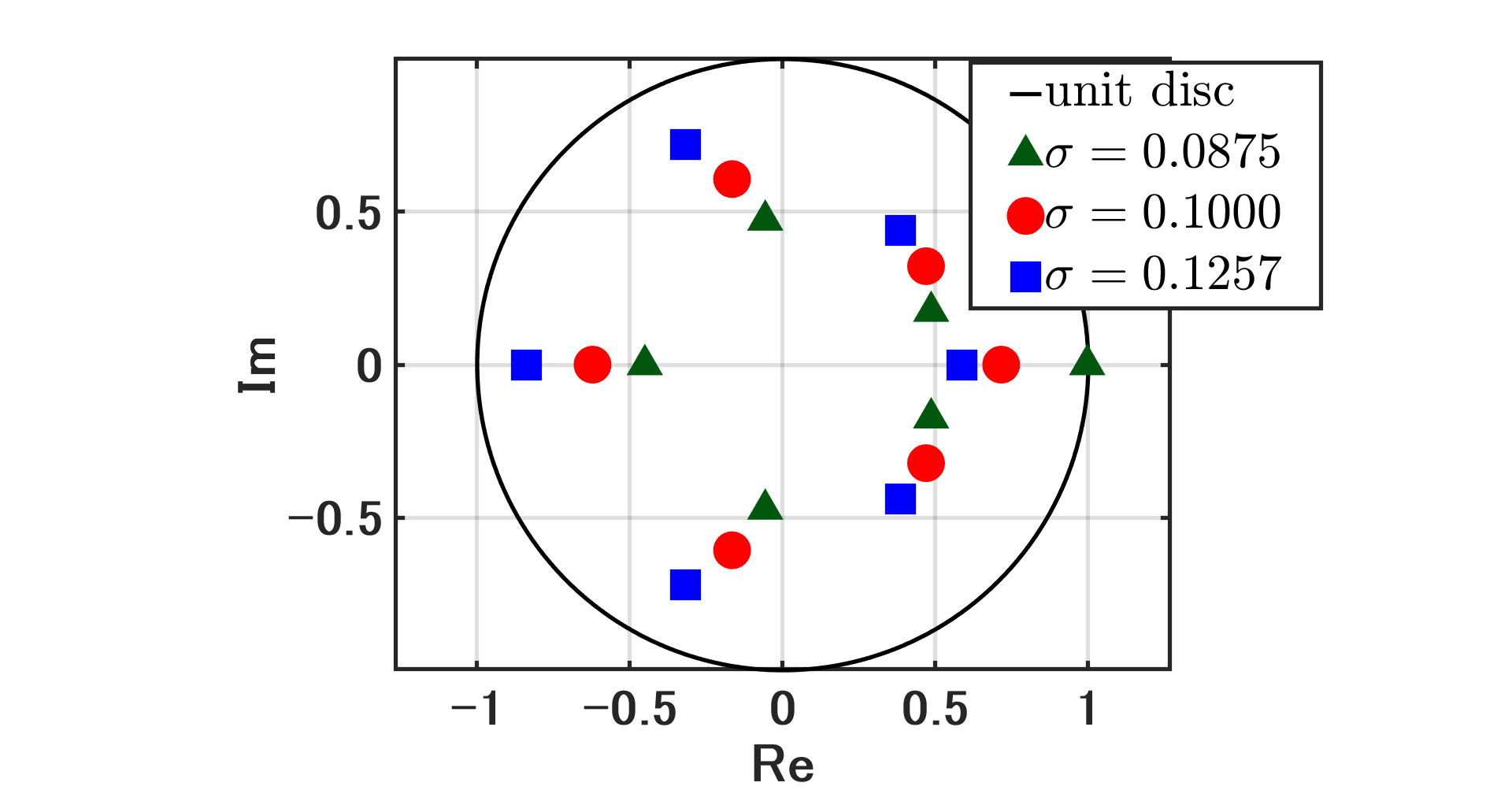}
\caption{Eigenvalues of closed-loop system for $A_s$and $B_s$\space$(\sigma=0.0875,\: 0.1000,\: 0.1257)$}\label{fig:eigenvalue_sigmacompare}
\end{centering}
\end{figure}

We conduct numerical simulations for these cases under the initial condition
\begin{equation*}
x_0 = \begin{bmatrix} 1 \\ -1\end{bmatrix},\ \:
u_{-1}=-1,\ u_{-2}=1,\ u_{-3}=-1,\ u_{-4}=1.
\end{equation*}
The simulation results are illustrated in Figure~\ref{fig:sig_compare}.  
For each value of $\sigma$, all the state variables are stabilized and decay to the origin.  
It can also be seen from the figure that the transient response of the closed-loop system strongly depends on the value of $\sigma$. 
To be specific, the convergence is very slow for $\sigma = 0.0875$ because of the closed-loop eigenvalue close to $1$.
As $\sigma$ increases, the convergence becomes faster. 
However, the transient response of the state variables is degraded and oscillatory for $\sigma=0.1257$, which is close to the upper bound of the feasible interval.

\begin{figure}[!bth]
\centering
\includegraphics[width=11.5cm]{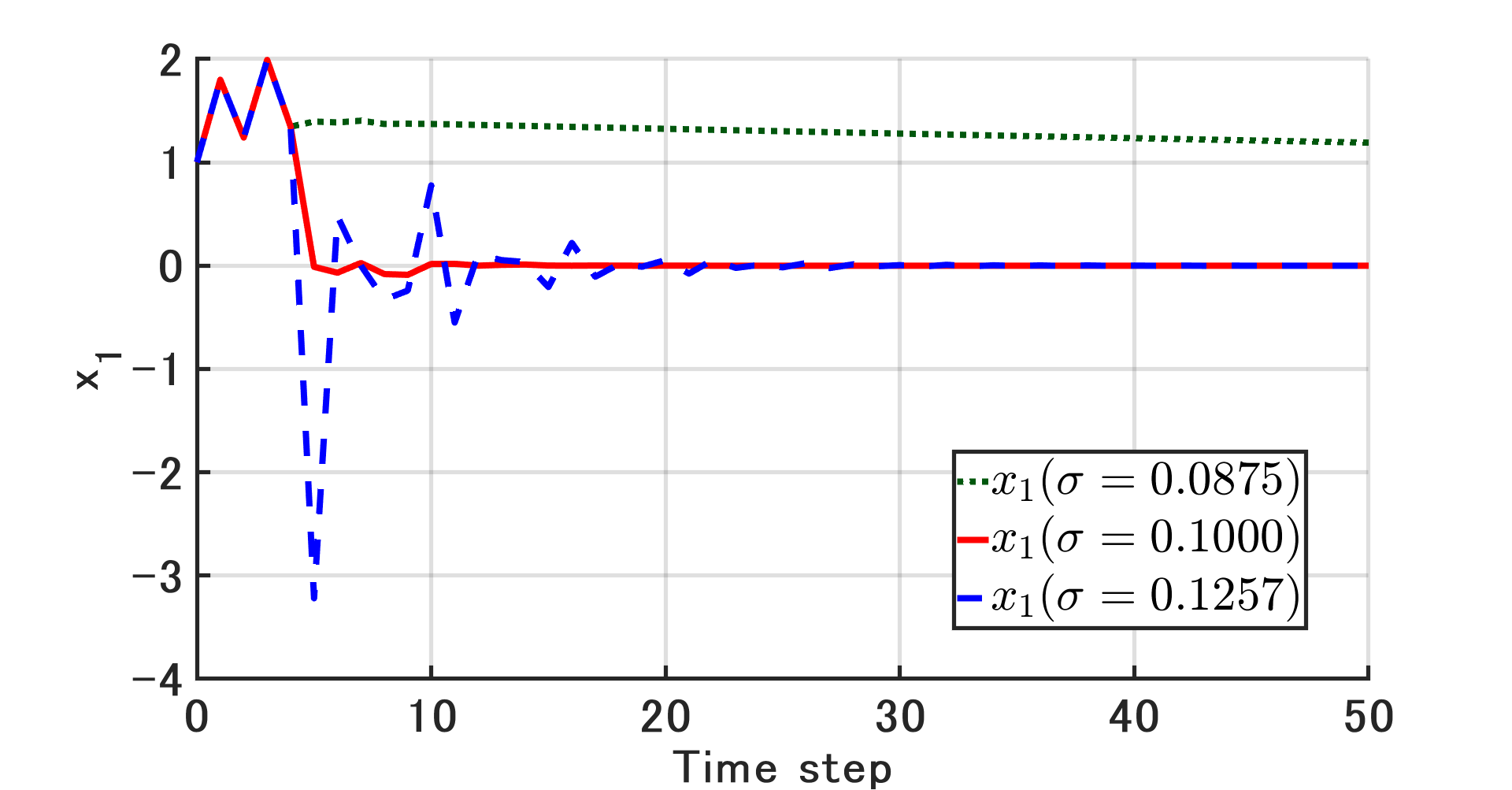}

\mbox{(a) state $x_{1t}$}

\medskip

\includegraphics[width=11.5cm]{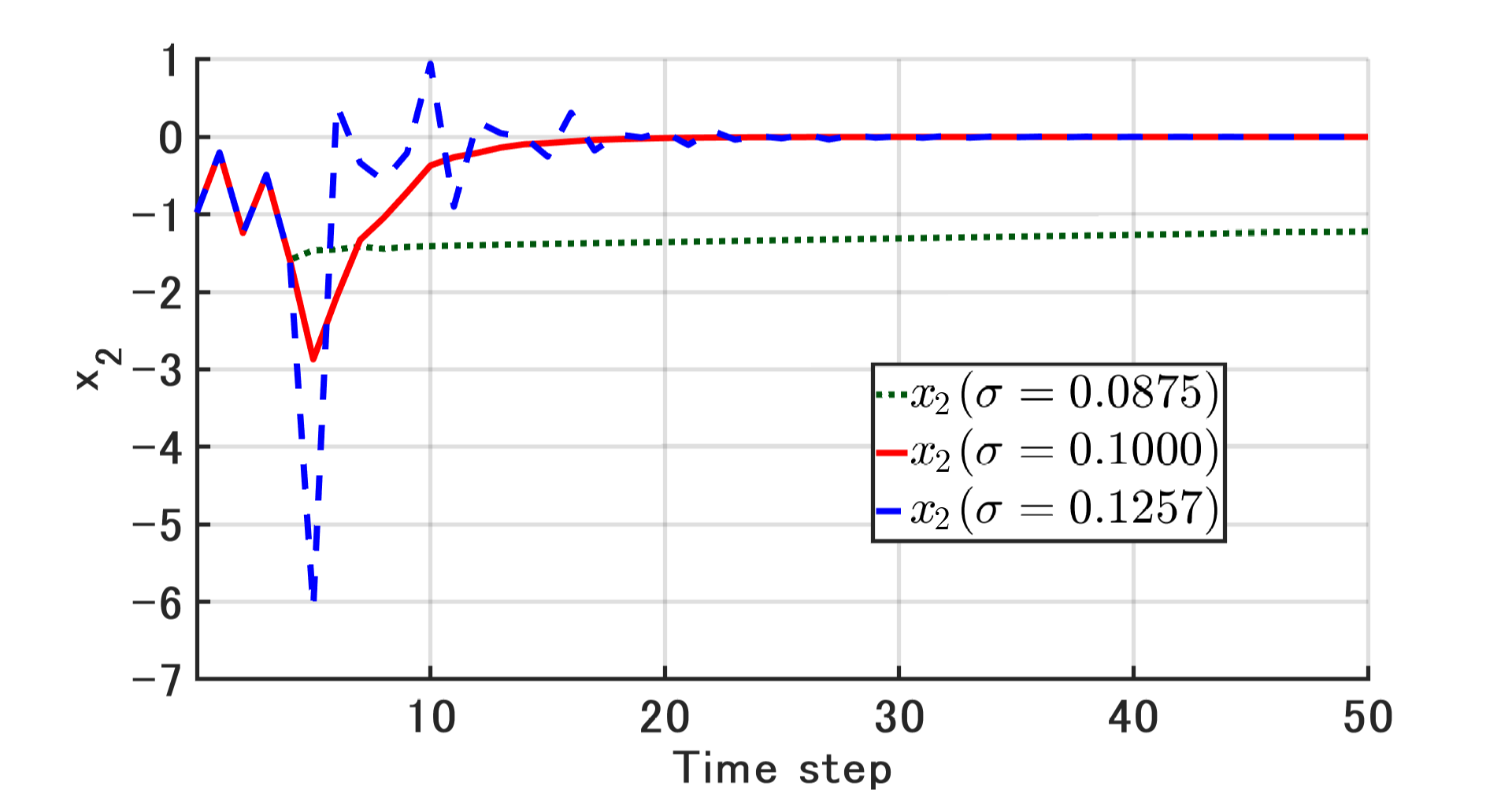}

\mbox{(b) state $x_{2t}$}

\includegraphics[width=11.5cm]{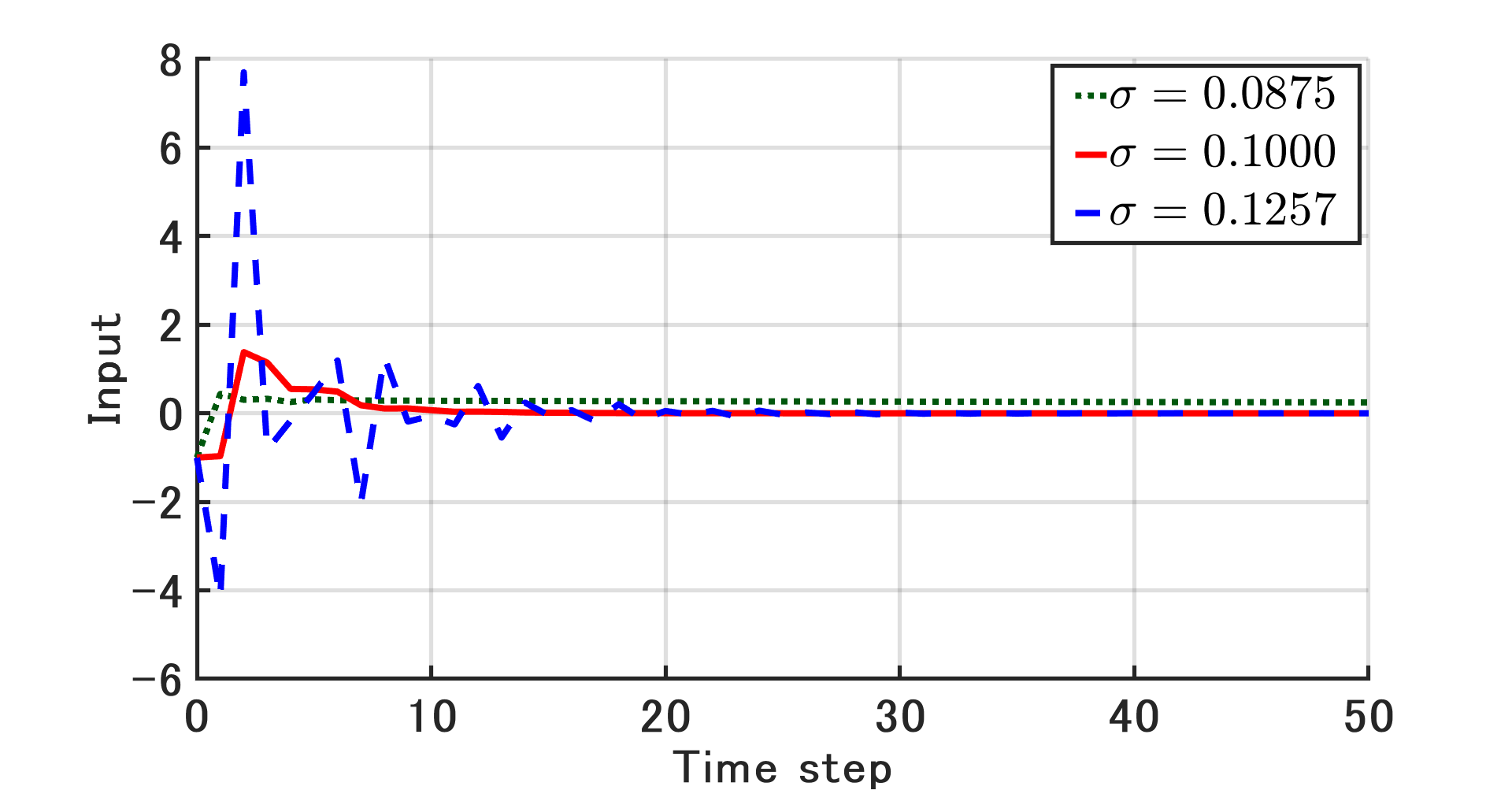}

\mbox{(c) input $u_t$}

\caption{Simulation Results for different values of $\sigma$}
\label{fig:sig_compare}
\end{figure}

The performance index values for the above simulations are summarized in Table~\ref{tab:performance_inequality}.
This table shows that the sub-optimal control performance is satisfied and that the value of $J$ deteriorates as $\sigma$ approaches either the upper or lower bound.

\begin{table}[!hbt]
\centering
\caption{Performance index values for different $\sigma$}
\label{tab:performance_inequality}
\begin{tabular}{cll}
\hline
$\sigma$ & Performance index $J$  & Upper bound $\gamma_{\rm opt}\|X_0\|^2$ \\
\hline
$0.0875$ & $6.788\times10^{-2} $& $\,\ \ \ 0.2470$ \\ 
$0.1000$    & $7.847\times10^{-3}$ & $\,\ \ \ 1.0296$ \\
$0.1257$ & $6.682\times10^{-2}$ & $116.0141$ \\ 
\hline
\end{tabular}
\end{table}

\begin{figure}[!htb]
\begin{centering}
\includegraphics[width=11.5cm]{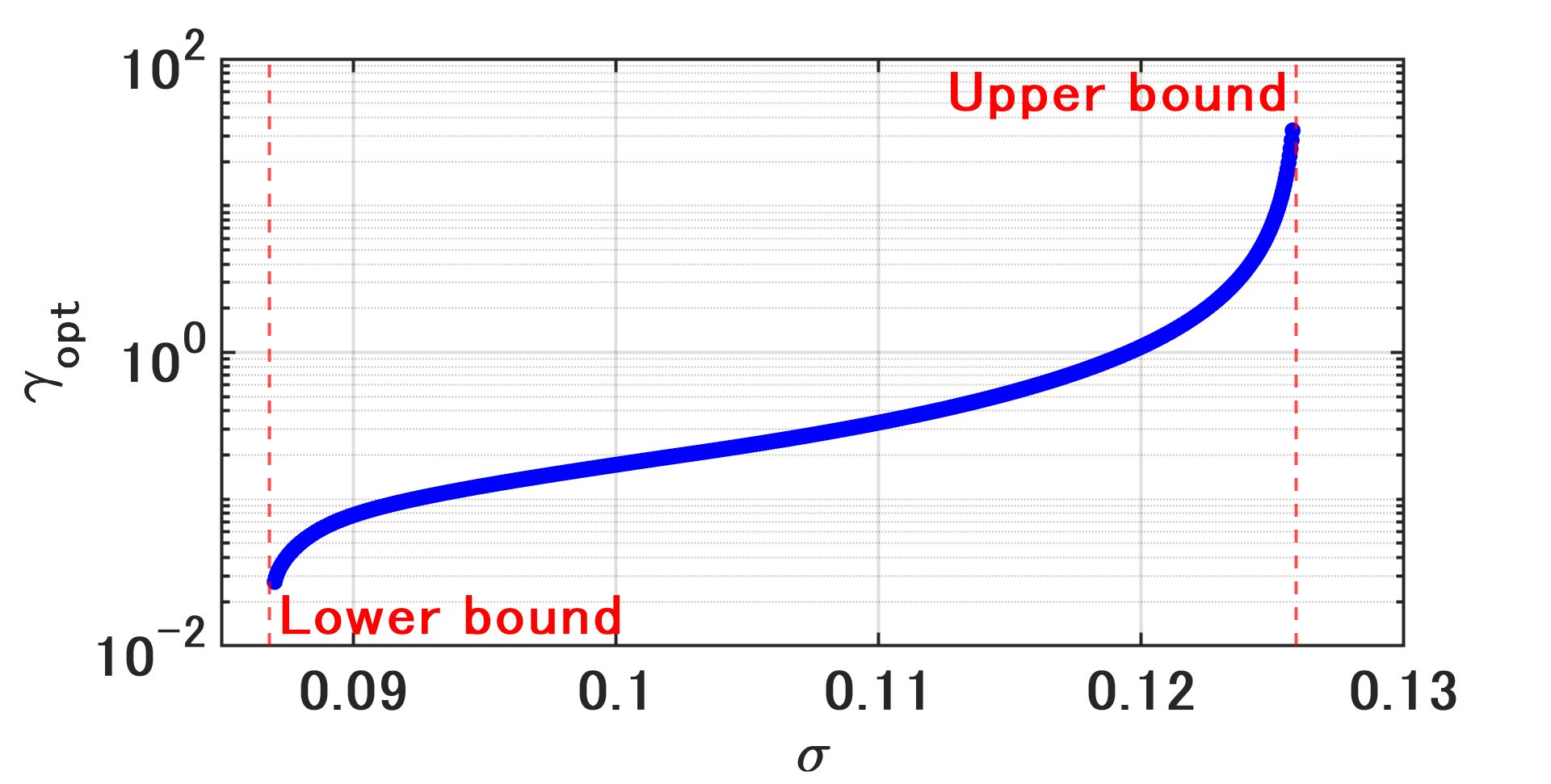}
\caption{Relation between $\sigma$ and achievable performance}
\label{fig:noise_bound}
\end{centering}
\end{figure}


Figure~\ref{fig:noise_bound} shows the change in the achievable performance level $\gamma_{\mathrm{opt}}$ with respect to $\sigma$.
It can be observed from the figure that $\gamma_{\mathrm{opt}}$ is monotone increasing with respect to $\sigma$.
This is because a larger value of $\sigma$ enlarges the set $\Sigma_{\mathcal{D}}$, and hence the stability and control performance must be guaranteed for a wider class of plants.
In particular, as $\sigma$ approaches the upper bound of the feasible interval,  robust stabilization of the plant set $\Sigma_{\mathcal{D}}$ becomes very difficult, and hence $\gamma_{\mathrm{opt}}$ diverges.
On the other hand, when $\sigma$ is small, $\Sigma_{\mathcal{D}}$ becomes smaller, which allows $\gamma_{\mathrm{opt}}$ to take a smaller value.

\subsection{Comparison with Previous Work~\cite{sicefes_ayaka}}

In the authors' previous work~\cite{sicefes_ayaka}, 
a necessary and sufficient condition for the informativity of noisy data for quadratic stabilization of a linear input-delay system.
In this sub-section, we compare the present result in Section~4 with the previous work~\cite{sicefes_ayaka} by synthesizing the stabilizing controllers with $\sigma = 0.1000$ and the same data set $\mathcal{D}$.
The feedback gain $K$ fot the previous work is obtained as
\begin{equation}
K =
\left[
\begin{array}{rrrrrr}
-6.0270 & -9.5519 & -10.6518 & -7.0103 & -4.3536 & -2.4217 \notag
\end{array}
\right].
\end{equation}
The closed-loop eigenvalues for this feedback gain are given by
\[
\{-0.1916,\,-0.4589 \pm 0.3387i,\,0.8277,\,0.1800 \pm 0.5877i\}.
\]

\begin{figure}[!thb]
\centering
\includegraphics[width=13.5cm]{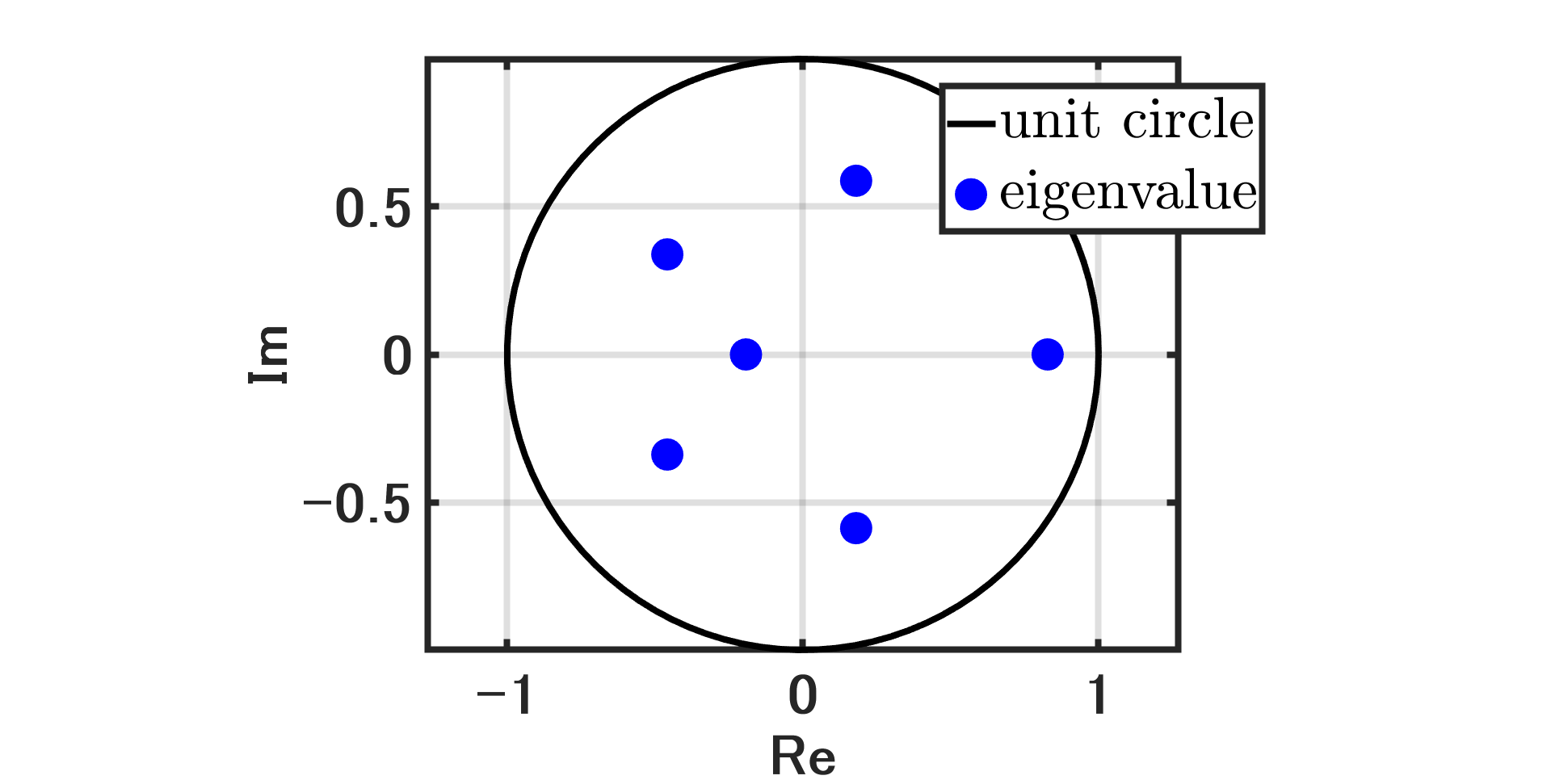}
\caption{Closed-loop eigenvalues for Ayaka \emph{et al.}(2025)}
\end{figure}

\begin{figure}[!thb]
\centering
\includegraphics[width=11.5cm]{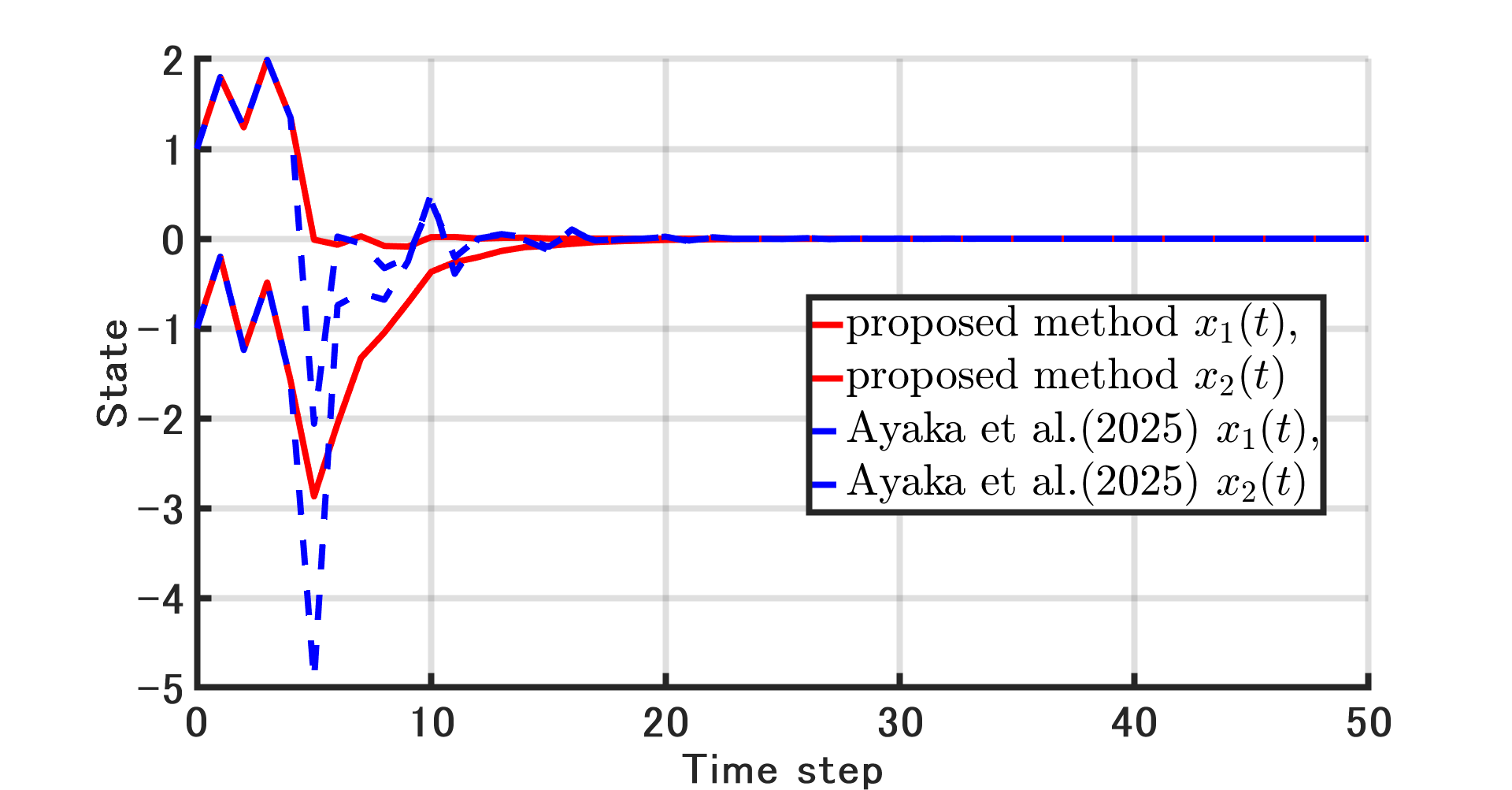}
\\ 

(a) State $x_t$

\medskip

\includegraphics[width=11.5cm]{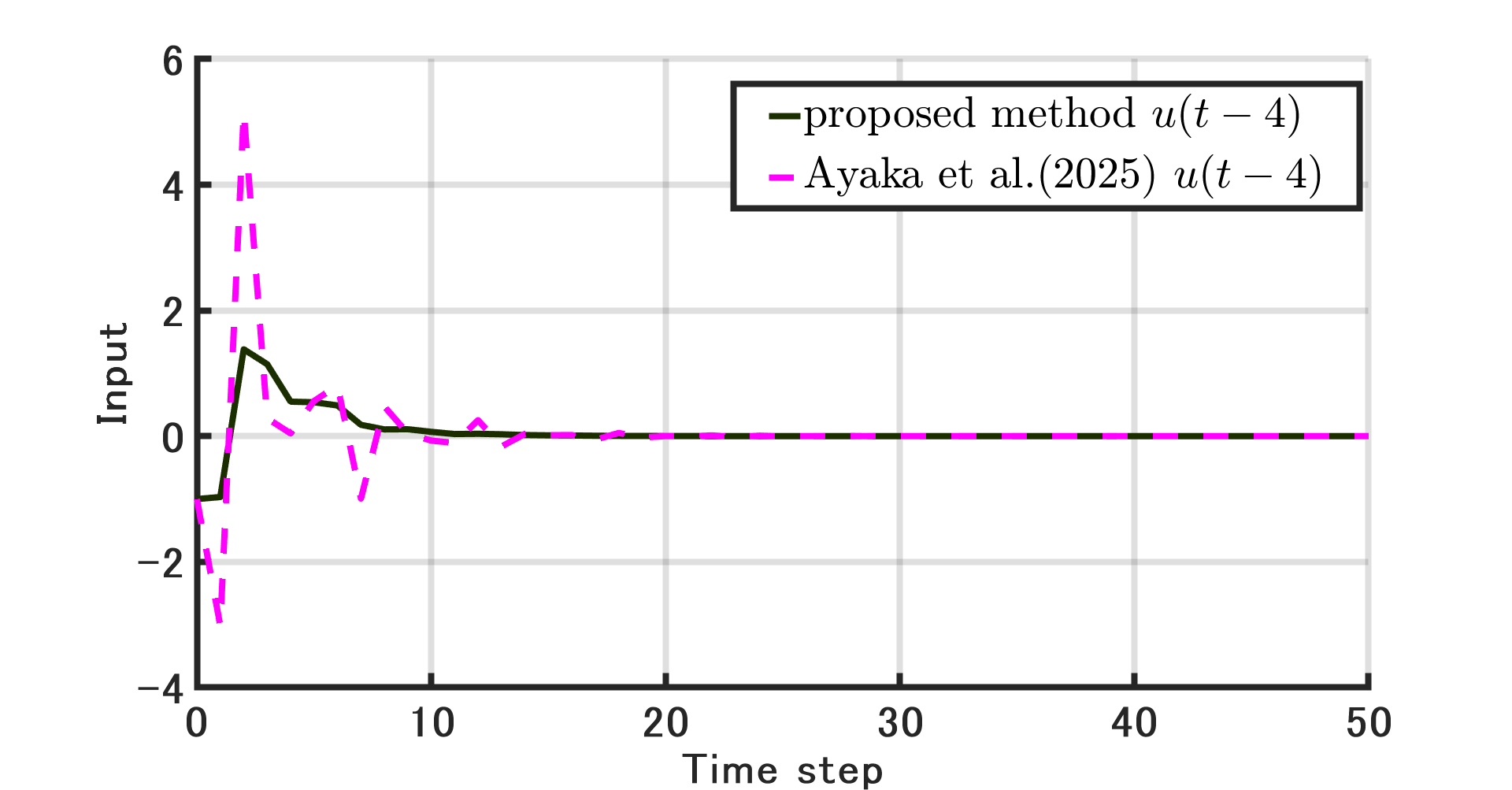}
\\

(b) Input $u_t$

\caption{Simulation results (comparison with the previous work~\cite{sicefes_ayaka})}
\label{fig:compare_previous}
\end{figure}

We conduct a simulation with $\sigma=0.1$ and the same initial values as in the previous sub-section. The proposed controller is synthesized with the same weights $Q_i$ $(i=0,1,\dots,4)$ and $R$. 
The resulting state and input trajectories are shown in Figure~\ref{fig:compare_previous}. 
It is observed from the figure that the proposed method outperforms in the transient responses of the state and input variables, in comparison with the previous method~\cite{sicefes_ayaka},
The resulting control performance from the previous method is $J= 3.3069\times 10^{-2}$, that is much larger than the performance achieved by the proposed method $J=7.847\times10^{-3}$.
Thus, we can say that the proposed method achieves suboptimal control with improved performance.

\section{Concluding Remarks}
\label{sect:6}
In this paper, we have proposed a novel data-driven method for sub-optimal LQR synthesis for linear input-delay systems, based on the informativity approach~\cite{sankou1}.
The proposed method enjoys the structure of input-delay systems and utilizes the noisy input-state data of the original plant \eqref{ids} rather than that of the augmented plant \eqref{eq:5}. This enables us to derive an LMI condition for the informativity of the noisy data for sub-optimal LQR synthesis in a compact form.
As verified through the numerical simulation, the proposed method improves the control performance in comparison with the authors' previous work~\cite{sicefes_ayaka} by achieving the LQ sub-optimality.
Moreover, the simulation results suggest that we need to carefully choose the parameter $\sigma$, which characterizes the plant models consistent with the noisy data, since it strongly affects the performance of the resulting closed-loop system,

As a future research direction, it remains to extend the present results of this paper to practical data-driven control designs for linear systems with delays, such as tracking control, remote control, networked control, and their experimental validations.

\section*{Notes on Contributors}
\paragraph*{Kohei Ayaka}
He received his B.Eng. from Ritsumeikan University, Japan in 2023. 
He is currently working toward his Master's degree in the Department of Electrical and Electronic Engineering, Ritsumeikan University. 
His research interests include data-driven control and its applications to delay systems. 
He is a student member of ISCIE.

\paragraph*{Takumi Namba}
He received his B.Eng., M.Eng., and Dr.Eng. degrees from Ritsumeikan University, Japan in 2018, 2020 and 2024, respectively. 
From April 2020 to August 2021, he was with Grid Aggregation Division, Toshiba Energy and Systems Solutions Corporation, Japan. 
From April 2024, he is an Assistant Professor at the Department of Electrical and Electronic Engineering, Ritsumeikan University. 
His current research interests include distributed optimization and control of multi-agent systems and their applications to energy management. 
He is a member of SICE, ISCIE, and IEEE.

\paragraph*{Kiyotsugu Takaba}
He received his B.Eng., M.Eng., and Dr.Eng. degrees all from Kyoto University, Japan, in 1989, 1991, and 1995, respectively. From 1991 to 1998, he was an Assistant Professor at the Department of Applied Mathematics and Physics, Kyoto University. From 1998 to 2012, he was an Associate Professor at the same department. 
In 2012, he joined the Department of Electrical and Electronic Engineering, Ritsumeikan University, where he is currently a Professor. 
His current research interests include the control and estimation of multi-agent systems and their applications. He is a member of ISCIE, IEEE, and SIAM.
\color{black}

\appendix

\section{Proof of Lemma~\ref{lem:expanded_S_lemma}}

(Necessity) Firstly, notice that (\ref{eq:8}) implies $P_c - Q_c > 0$. Then, since $S_N$ is a closed set, there exists a scalar constant $\varepsilon'>0$ satisfying 
\begin{align}
\label{slemmaproof1}
    \begin{bmatrix}
P'_a & {P'_b}^\top \\
P'_b & P'_c
\end{bmatrix}
>
\begin{bmatrix}
Z^\top Q_a Z & Z^\top Q_b^\top \\
Q_b Z & Q_c+\varepsilon'I_\ell
\end{bmatrix}
\quad \forall Z \in \mathcal{S}_N
.
\end{align}
Noting $P_c - Q_c-\varepsilon'I_\ell>0$, application of the Schur complement formula yields
\begin{align*}
& 
\begin{bmatrix}
P'_a & {P'_b}^\top \\
P'_b & P'_c
\end{bmatrix}
> 
\begin{bmatrix}
Z^\top Q_a Z & Z^\top Q_b^\top \\
Q_b Z & Q_c+\varepsilon'I_\ell
\end{bmatrix}\\
& \Leftrightarrow P'_a-Z^\top Q_a Z -(P'_b-Q_bZ )^\top(P'_c-Q_c-\varepsilon'I_\ell)^{-1}(P'_b-Q_bZ )
> 0 \\
&\Leftrightarrow
\begin{bmatrix}
I_n \\
Z
\end{bmatrix}^\top
\begin{bmatrix}
P'_a - {P'_b}^\top \Xi^{-1} P'_b & P'_b \Xi^{-1} Q_b^\top \\
Q_b \Xi^{-1} {P'_b}^\top & -Q_a - Q_b \Xi^{-1} Q_b^\top
\end{bmatrix}
\begin{bmatrix}
I_n \\
Z
\end{bmatrix} \notag
> 0
, 
\end{align*}
where we have defined $\Xi=P'_c - Q_c-\varepsilon'I_\ell >0$. Thus, (\ref{slemmaproof1}) is equivalent to
\begin{align}
\label{eq.Xi_lemma}
\begin{bmatrix}
I_n \\
Z
\end{bmatrix}^\top
\begin{bmatrix}
P'_a - {P'_b}^\top \Xi^{-1} P'_b & P'_b \Xi^{-1} Q_b^\top \\
Q_b \Xi^{-1} {P'_b}^\top & -Q_a - Q_b \Xi^{-1} Q_b^\top
\end{bmatrix}
\begin{bmatrix}
I_n \\
Z
\end{bmatrix} 
> 0 \quad \forall Z \in \mathcal{S}_N.
\end{align}
We get $-Q_a - Q_b^\top \Xi^{-1} Q_b \le 0$ from $\Xi > 0$ and the assumption $Q_a> 0$. 
It then follows from the matrix $S$-lemma (Theorem 13 in \cite{sankou2}) that (\ref{eq.Xi_lemma}) holds if and only if there exist constants $\varepsilon >0$ and $\alpha \ge 0$ satisfying
\begin{align*}
\begin{bmatrix}
P'_a - {P'_b}^\top \Xi^{-1} P'_b-\varepsilon I_m & P'_b \Xi^{-1} Q_b^\top \\
Q_b \Xi^{-1} {P'_b}^\top & -Q_a - Q_b \Xi^{-1} Q_b^\top
\end{bmatrix}
\ge \alpha
\begin{bmatrix}
N_a & N_b^\top \\
N_b & N_c
\end{bmatrix}.
\end{align*}
By applying the Schur complement formula again to this inequality, we obtain the matrix inequality of (\ref{eq:9}). This completes the proof of the necessity.

\medskip

(Sufficiency)
Assume that there exist positive constants $\varepsilon, \varepsilon'$ and a non-negative constant $\alpha$ satisfying (\ref{eq:9}). 
Pre- and post-multiplying (\ref{eq:9}) by 
{\small $\begin{bmatrix}
I_n & 0 \\
Z & 0 \\
0 & I_\ell
\end{bmatrix}^\top$} and its transpose, respectively, yield
\begin{align*}
\begin{bmatrix}
P'_a & {P'_b}^\top\\
P'_b & P'_c
\end{bmatrix}
-\begin{bmatrix}
Z^\top Q_a Z & Z^\top Q_b^\top \\
Q_b Z & Q_c
\end{bmatrix}
\ge \alpha
\begin{bmatrix}
\begin{bmatrix}
I_n\\Z    
\end{bmatrix}^\top N \begin{bmatrix}
I_n\\Z    
\end{bmatrix} & 0 \\
0 & 0
\end{bmatrix}+
\begin{bmatrix}
\varepsilon I_n & 0 \\
0 & \varepsilon' I_\ell
\end{bmatrix}
.\end{align*}
Since $\varepsilon, \varepsilon'>0$ and $\alpha \ge 0$, this inequality implies that
\begin{align*}
\begin{bmatrix}
P'_a & {P'_b}^\top\\
P'_b & P'_c
\end{bmatrix}
-\begin{bmatrix}
Z^\top Q_a Z & Z^\top Q_b^\top \\
Q_b Z & Q_c
\end{bmatrix}
\ge
\begin{bmatrix}
\varepsilon I_n & 0 \\
0 & \varepsilon' I_\ell
\end{bmatrix}
> 0
\end{align*}
holds for every $Z\in \mathcal{S}_N$. This completes the proof of sufficiency.

\end{document}